\let\accentvec\vec    		  	
\documentclass[runningheads,envcountsame]{llncs}
\let\vec\accentvec    		  	

\usepackage[utf8x]{inputenc}
\usepackage{amsmath,qsymbols,mathpartir}
\usepackage{bussproofs,prooftree}
\usepackage{tikz}

\DeclareFontFamily{U}{stmry}{}
\DeclareFontShape{U}{stmry}{b}{n}
   {  <5> <6> <7> <8> <9> <10> gen * stmary
      <10.95><12><14.4><17.28><20.74><24.88>stmary10%
   }{}
\DeclareFontShape{U}{stmry}{m}{n}
   {  <5> <6> <7> <8> <9> <10> gen * stmary
      <10.95><12><14.4><17.28><20.74><24.88>stmary10%
   }{}

\newcommand{\notation}{\noindent\textbf{Notation:} }	
\newcommand{\ie}{{i.e.\ }}			
\newcommand{\eg}{{e.g.\ }}			
\newcommand{\cf}{{cf.~}} 	   		
\newcommand{\cbv}{\textsc{cbv}}                 
\newcommand{\cbn}{\textsc{cbn}}                 

\renewcommand{\t}{\ensuremath{\mathbf{t}}}      
\renewcommand{\u}{\ensuremath{\mathbf{u}}}      
\newcommand{\s}{\ensuremath{\mathbf{s}}}        
\renewcommand{\v}{\ensuremath{\mathbf{v}}}      
\newcommand{\zero}{{\ensuremath{\mathbf{0}}} }	
\newcommand{\subs}[2][x]{\ensuremath{\{#2/#1\}}}
\newcommand{\fv}[1]{\ensuremath{f\!v}(#1)}      

\newcommand{\T}{\ensuremath{T}}                 
\newcommand{\R}{\ensuremath{R}}                 
\renewcommand{\S}{\ensuremath{S}}                 
\newcommand{\U}{\ensuremath{U}}                 
\newcommand{\V}{\ensuremath{V}}                 
\newcommand{\W}{\ensuremath{W}}                 
\newcommand{\add}{\ensuremath{Additive}}	
\newcommand{\tzero}{{\ensuremath{\overline{0}}} } 
\newcommand{\sui}[1]{\sum_{i=1}^{#1}}		
\newcommand{\suj}[1]{\sum_{j=1}^{#1}}		
\newcommand{\type}{\colon\!}			
\newcommand{\ve}[1]{\ensuremath{\mathbf{#1}}}	
\newcommand{\thesi}{\vdash_{\scriptstyle F}}	
\newcommand{\uno}{\ensuremath{\mathbf{1}}}	
\newcommand{\trad}[1]{|#1|}			
\newcommand{\ssub}{\prec}			
\newcommand{\tsubs}[2][X]{\ensuremath{[#2/#1]}}	

\renewcommand{\b}{\ensuremath{\mathbf{b}}}      
\newcommand{\cre}{\ensuremath{\mathcal{R\!C}}}	
\newcommand{\terms}{\ensuremath{`L_0}}		
\newcommand{\sn}{\ensuremath{S\!N_0}}		
\newcommand{\neu}{\ensuremath{\mathcal{N}}}	
\newcommand{\sets}{\ensuremath{\mathcal{S}}}	
\newcommand{\red}[2][]{\ensuremath{\mathit{Red_{#1}(#2)}}} 
\newcommand{\cru}{\ensuremath{(\mathit{CR}_1)}}	
\newcommand{\crd}{\ensuremath{(\mathit{CR}_2)}}	
\newcommand{\crt}{\ensuremath{(\mathit{CR}_3)}}	
\newcommand{\clo}[1]{\ensuremath{\overline{#1}}}
\newcommand{\mas}{\ensuremath{\mp}}		
\newcommand{\itp}[2][`r]{\ensuremath{\llbracket #2 \rrbracket_{#1}}} 
\newcommand{\valid}{\vDash}			
\newcommand{\satis}[3][`r]{\ensuremath{#2`:\itp[#1]{#3}}} 

\newcommand{\F}{\ensuremath{\text{System}~F_{\!P}}}
\newcommand{\tof}{\ensuremath{"=>"}}			
\newcommand{\sadd}{\ensuremath{Add_{str}}}		
\newcommand{\A}{\ensuremath{\mathcal{T}}}		
\newcommand{\Az}{\ensuremath{\mathtt{Z}}}		
\newcommand{\As}{\ensuremath{\mathtt{S}}}		
\newcommand{\lab}{\ensuremath{s}}                       
\newcommand{\estruct}{\ensuremath{\to_{E'}}}		
\newcommand{\pair}[2]{\ensuremath{\langle #1,#2\rangle}}
\newcommand{\pil}[1]{\ensuremath{`p_{\mathtt{l}}(#1)}}        
\newcommand{\pir}[1]{\ensuremath{`p_{\mathtt{r}}(#1)}}        
\newcommand{\tradt}[2][\mathcal{D}]{[#2]_{#1}}		
\newcommand{\tradti}[1]{\langle\!|{#1}|\!\rangle}       
\newcommand{\tradi}[1]{(\!|{#1}|\!)}   			
\newcommand{\D}{\ensuremath{\mathcal{D}}}		

\usetikzlibrary{positioning}
\usetikzlibrary{shapes}
\tikzset{
  every picture/.style={every node/.style={anchor=mid}},
  to right/.style= {                         
    to path={ (\tikztostart.mid east)
              --
              (\tikztotarget.mid west)
              \tikztonodes }
          },
  to left/.style={                         
    to path={ (\tikztostart.mid west)
              --
              (\tikztotarget.mid east)
              \tikztonodes }
          },
  m/.style       = {execute at begin node= \small$,
                    execute at end node=$%
                  },
  tiny m/.style  = {execute at begin node= $\scriptstyle,
                    execute at end node=$%
},
  etoile/.style  = {at end,            
                    font=$\scriptstyle *$,
                    auto=left},
  punto/.style={fill,
                draw,
                circle,
                inner sep=1pt,
                label=below:{$#1$}
              },
  punto/.default={},
  triang/.style={isosceles triangle,
                 draw,
                 shape border rotate=90,
                 anchor=left corner,
                 minimum width=#1
                 },
  style={
         parent anchor=center,
         child anchor=north,
         sibling distance=3em,
         level distance=2em}
     }

\begin{document}

\title{Linearity in the Non-deterministic\\ Call-by-Value Setting}
\author{Alejandro D\'iaz-Caro\inst{1,}\thanks{Supported by grants from DIGITEO and R\'egion \^Ile-de-France}
  \and Barbara Petit\inst{2}}
\institute{Universit\'e Paris 13, Sorbonne Paris Cit\'e, LIPN, F-93430, Villetaneuse, France
  \and \textsc{Focus} (\textsc{inria}) -- Universit\`a di Bologna, Italy
  }

\maketitle

\begin{abstract}
We consider the non-deterministic extension of the call-by-value lambda calculus, which corresponds to the additive fragment of the linear-algebraic lambda-calculus. We define a fine-grained type system, capturing the right linearity present in such formalisms. After proving the subject reduction and the strong normalisation properties, we propose a translation of this calculus into the System $F$ with pairs, which corresponds to a non linear fragment of linear logic. The translation provides a deeper understanding of the linearity in our setting.
\end{abstract}

\section*{Introduction}
\label{sec:intro}
Several non-deterministic extensions of $\lambda$-calculus have been proposed in the literature, \eg \cite{BoudolIC94,BucciarelliEhrhardManzonettoAPAL12,DezaniciancagliniDeliguoroPipernoTCS96,DezaniciancagliniDeliguoroPipernoSIAM98}. In these approaches, the sometimes called {\em must-convergent parallel composition}, is such that if $\t$ and $\u$ are two $\lambda$-terms, $\t+\u$ (also written $\t\parallel\u$) represents the computation that runs either $\t$ or $\u$ non-deterministically. Therefore, $(\t+\u)\s$ can run either $\t\s$ or $\u\s$, which is exactly what $\t\s+\u\s$ expresses. Extra rewriting rules (or equivalences, depending on the presentation) are set up to account for such an interpretation, \eg $(\t+\u)\s\to\t\s+\u\s$.

This right distributivity can alternatively be seen as the one of the function sum: $(\mathbf{f}+\mathbf{g})(x)$ is defined as $\mathbf{f}(x)+\mathbf{g}(x)$. This is the approach of the algebraic lambda-calculi presented in~\cite{ArrighiDowekRTA08} and~\cite{VauxMSCS09}, that were introduced independently but that resulted afterwards to be strongly related~\cite{AssafPerdrixDCM11,DiazcaroPerdrixTassonValironHOR10}. In these algebraic calculi, a scalar pondering each `choice' is considered in addition to the sum of terms.

In the call-by-value (or \cbv) version of these algebraic/non-deterministic calculi, \eg \cite{BoudolIC94,DezaniciancagliniDeliguoroPipernoSIAM98,ArrighiDowekRTA08}, it is natural to consider also the left distributivity of application over sums: $\t(\u+\s)\to\t\u+\t\s$. To our knowledge, this was first observed in~\cite{HennessySIAM80}. Indeed, a sum $\u+\s$ is not a value, in the sense that it represents a non-deterministic choice that remains to be done, and therefore cannot subsitute the argument~$x$. In algebraic terms, it means that functions are linear: $\mathbf{f}(x+y)=\mathbf{f}(x)+\mathbf{f}(y)$.

The work we present here is motivated by a better understanding of this linearity, and so our first attempt was to interpret such a \cbv\ calculus in Linear Logic~\cite{GirardTCS87} (indeed linear functions can be precisely characterised in this logic). Surprisingly, it appeared that the target calculus was a non linear fragment of the intuitionistic multiplicative exponential Linear Logic (\textsc{imell}), shining a light on the difference between the linearity in these non-deterministic calculi, and the common algebraic notion of linear functions. Since the non linear fragment of \textsc{imell} corresponds to the System~$F$ with pairs~\cite[Sec.~1.5]{Dicosmo95}, and this latter might be better known by the reader, we present in this paper a (reversible) translation into the System~$F$ with pairs.

Notice also that the left distributivity of application over sum induces a completely different computational behaviour compared to the one in \cbn\ calculi. Consider for instance the term $`d=`lx.xx$ applied to a sum $\t+\u$. In the first case, it reduces to $`d\t+`d\u$ and then to $\t\t+\u\u$, whereas a \cbn\ reduction would lead to $(\t+\u)(\t+\u)$ and then to $\t(\t+\u)+\u(\t+\u)$. In particular, the \cbv\ algebraic calculus we mentioned above (Lineal,~\cite{ArrighiDowekRTA08}) was originally meant to express quantum computing, where a superposition $\t+\u$ is seen as a quantum superposition. Hence reducing $`d(\t+\u)$ into $(\t+\u)(\t+\u)$ is considered as the forbidden quantum operation of ``cloning''~\cite{WoottersZurekNATURE82}, while the alternative reduction to $\t\t+\u\u$ is seen as a ``copy'', or \textsc{cnot}, a fundamental quantum operation~\cite{MonroeMeekhofKingItanoWinelandPRL95}.

\paragraph{Outline.}
In this paper we propose (in Sec.~\ref{sec:add}) a type system, called \add, capturing the linear \cbv\ behaviour of the sum operator that we discussed above. Then we prove its correctness properties, namely subject reduction and strong normalisation in Sec.~\ref{sec:sr-sn}. Its logical interpretation (that is, the translation into System $F$ with pairs) is developed in Sec.~\ref{sec:trans}. We conclude with a discussion about the linearity of the call-by-value setting. We leave in the appendices extra examples and some technical details such as auxiliary lemmas.

\section{The Calculus}
\label{sec:add}
\subsection{The Language}\label{sec:language}

We consider the call-by-value $\lambda$-calculus~\cite{PlotkinTCS75} extended with a non-deterministic operator in the spirit of the parallel composition from~\cite{BucciarelliEhrhardManzonettoAPAL12}. This setting can be seen as the additive fragment of Lineal~\cite{ArrighiDowekRTA08}. The set of {\em terms} and the set of {\em values} are defined by mutual induction as follows (where variables range over a countable set and are denoted by $x,y,z$):
\vspace{-0.25cm}

$$\begin{array}{l@{\hspace{1cm}}r@{\quad::=\quad}l}
 \mbox{Terms:}  & \t,\u,\s & \v~|~\t\u~|~\t+\u~|~\zero\\
 \mbox{Values:} & \v & x~|~\lambda x.\t
\end{array}$$

Intuitively $\t+\u$ denotes the {\em non-deterministic choice} between $\t$ and $\u$, and hence, as discussed in the introduction, $(\t+\u)\s$ {\em reduces} to the non-deterministic choice $\t\s+\u\s$. Analogously, in this call-by-value setting, $\t(\u+\s)$ reduces to $\t\u+\t\s$. The term $\zero$ is introduced to express the {\em impossible computation}, and hence $\t+\zero$ always reduces to $\t$, while $\t\zero$ and $\zero\t$ reduce to $\zero$, because none of them continue reducing (notice that $\zero$ is not a value), and have an impossible computation on them.
Since the operator $+$ represents a non deterministic choice, where no one have precedence, terms are considered modulo associativity and commutativity of~$+$ (that is an {\em AC-rewrite system}~\cite{JouannaudKirchnerSIAM86}). Notice that considering $\t+\u$ either as a sum of functions or as a sum of arguments---depending on its position---is also natural with the previous definitions, where $\zero$ becomes the sum of $0$ elements.

The {\em $\alpha$-conversion} and the set \fv{\t} of {\em free variables of $\t$} are defined as usual (cf.~\cite[Sec.~2.1]{Barendregt84}). We say that a term $\t$ is closed whenever $\fv{\t}=\emptyset$. Given a term $\t$ and a value $\v$, we denote by $\t\subs{\v}$ to the term obtained by simultaneously substituting $\v$ for all the free occurrences of $x$ in $\t$, taking care to rename bound variables when needed in order to prevent variable capture. Hereafter, terms are considered up to $\alpha$-conversion.
The five rewrite rules plus the $\beta$-reduction are summarised as follows.
\medskip

\noindent\begin{tabular}{c@{\hspace{1,25cm}}c@{\qquad}cc@{\hspace{1,25cm}}c}
\emph{Distributivity rules:} & \multicolumn{2}{c}{\emph{Zero rules:}} && \emph{$`b$-reduction:}\\
$(\t+\u)\s\to\t\s+\u\s$,	& $\zero\t\to \zero$, & $\t + \zero\to \t$, && $(\lambda x.\t)\v\to\t\subs{\v}$.\\
$\t(\u+\s)\to\t\u+\t\s$,	& $\t\zero\to \zero$, &	&&
\end{tabular}

\subsection{The \add\ Type System}\label{sec:types}

Our objective is to define a type system, capturing as much as possible the behaviour of $+$. Roughly speaking, we want a system where, if $\t$ has type $\T$ and $\u$ has type $\R$, then $\t+\u$ has type $\T+\R$. So the natural typing rule for such a construction is ``$\Gamma\vdash\t:\T$ and $\Gamma\vdash\u:\R$ entails $\Gamma\vdash\t+\u:\T+\R$''. We also want a special type distinguishing the impossible computation $\zero$, which we call $\tzero$. Due to the associative and commutative nature of $+$, we consider an equivalence between types taking into account its commutative nature. Hence if $\T+\R$ is a type, $\R+\T$ is an equivalent type. Also the neutrality of $\zero$ with respect to $+$ is captured by an equivalence between $\T+\tzero$ and $\T$. Finally, as usual the arrow type $\T\to\R$ characterises the functions taking an argument in $\T$ and returning an element of $\R$. However, notice that the type $(\T+\R)\to\S$ captures a behaviour that is not appearing in our setting: there is no function taking a non-deterministic superposition as argument. Indeed, if $\v_1$ has type $\T$ and $\v_2$ type $\R$, any function $\t$ distributes $\t(\v_1+\v_2)$ as $\t\v_1+\t\v_2$, so $\t$ needs to be characterised by a function taking both $\T$ and $\R$, but not simultaneously. In order to capture such a behaviour, we introduce a {\em unit} type $\U$ (i.e.~an atomic type with respect to $+$), capturing elements which are not sums of elements, and hence the arrow types have the shape $\U\to\T$, where the different arguments to which the function can be applied, are captured by polymorphic types with variables ranging on unit types. For example, the previous term $\t$ can have type $\forall X.(X\to\S)$, where if $\t$ is applied to the above discussed $\v_1+\v_2$ of type $\T+\R$, it reduces to $\t\v_1+\t\v_2$ of type $\S\tsubs{\T}+S\tsubs{\R}$.

To take into account the above discussion, the grammar of the \add\ type system is defined by mutual induction as follows (where type variables range over a countable set and are denoted by $X,Y,Z$):
\vspace{-0.25cm}

$$\begin{array}{l@{\hspace{1cm}}r@{\quad::=\quad}l}
 \mbox{Types:}  	& \T,\R,\S & \U~|~\T+\R~|~\tzero\\
 \mbox{Unit types:} & \U,\V,\W & X~|~\U\to\T~|~\forall X.U
\end{array}$$

Contexts are denoted by $\Gamma,\Delta$ and are defined as sets of pairs $x:\U$, where each term variable appears at most once.
The substitution of~$X$ by $\U$ in~$\T$ is defined analogously to the substitution in terms, and is written~$\T\tsubs{\U}$. We also use the vectorial notation~$\T\tsubs[\vec X]{\vec{\U}}$ for $\T\tsubs[X_1]{\U_1}\cdots\tsubs[X_n]{\U_n}$ if $\vec{X}=X_1,\dots,X_n$ and $\vec{\U}=\U_1,\dots,\U_n$. To avoid capture, we consider that $X_i$ cannot appear free in $\U_j$, with $j<i$. Free and bound variables of a type are assumed distinct.

The above discussed equivalence relation~$\equiv$ on types, is defined as the least congruence such that:

\hfill
$\begin{array}{c@{\hspace{1cm}}c@{\hspace{1cm}}c}
\T+\R\equiv\R+\T,	&	\T+(\R+\S)\equiv(\T+\R)+\S,	&	\T+\tzero\equiv\T.
  \end{array}$\hfill\

Within this equivalence, it is consistent to use the following notation:

\notation
$\sui{0}{\T} = \tzero$\quad;\quad
$\sui{\alpha}{\T_i} = \sui{\alpha-1}{\T_i} + \T_\alpha$\quad
if $\alpha\geq 1$.

\begin{remark}
  \label{rem:unitsaddition}
  Every type is equivalent to a sum of unit types.
\end{remark}

Returning to the previous example, $\t(\v_1+\v_2)$ reduces to $\t\v_1+\t\v_2$ and its type have to be an arrow with a polymorphic unit type at the left. Such a type must allow to be converted into both the type of $\v_1$ and the type of $\v_2$. Hence, consider $\V_1$ and $\V_2$ to be the respective types of $\v_1$ and $\v_2$, we need $\t$ to be of type $\forall X.(\U\to \S)$ for some $\S$ and where $\U\tsubs{\W_1}=\V_1$ and $\U\tsubs{\W_2}=\V_2$ for some unit types $\W_1$ and $\W_2$. That is, we need that if $\t$ has such a type, then $\v_1$ has type $\U\tsubs{\W_1}$ and $\v_2$ type $\U\tsubs{\W_2}$. We can express this with the following rule
$$\prooftree\Gamma\vdash\t:\forall X.(\U\to \S)\qquad\Gamma\vdash\v_1+\v_2:\U\tsubs{\W_1}+\U\tsubs{\W_2}
\justifies\Gamma\vdash\t(\v_1+\v_2):\S\tsubs{\W_1}+\S\tsubs{\W_2}
\endprooftree$$

In the same way, for the right distributivity, if $\t$ and $\u$ are two functions of types $\U\to\T$ and $\V\to\R$ respectively, then the application $(\t+\u)\v$ needs $\U$ and $\V$ to be the type of $\v$. Therefore, the polymorphism plays a role again, and if $\t$ has type $\forall X.(\U\to\T)$ and $\u$ has type $\forall X.(\V\to\R)$ such that $\U\tsubs{\W_1}=\V\tsubs{\W_2}$ and also equal to the type of $\v$, then $(\t+\u)\v$ has a type. It can be expressed by
$$\prooftree\Gamma\vdash\t+\u:\forall X.(\U\to\S)+\forall X.(\V\to\R)\qquad\Gamma\vdash\v:\U\tsubs{\W_1}=\V\tsubs{\W_2}
\justifies\Gamma\vdash(\t+\u)\v:\S\tsubs{\W_1}+\R\tsubs{\W_2}
\endprooftree$$

Notice that when combining both cases, for example in $(\t+\u)(\v_1+\v_2)$, we need the type of $\t$ to be an arrow accepting both the type of $\v_1$ and the type of $\v_2$ as arguments, and the same happens with the type of $\u$. So, the combined rule is
$$\prooftree\Gamma\vdash\t+\u:\forall X.(\U\to\S)+\forall X.(\U\to\R)\qquad\Gamma\vdash\v_1+\v_2:U\tsubs{V}+U\tsubs{W}
\justifies\Gamma\vdash(\t+\u)(\v_1+\v_2):\S\tsubs{V}+\R\tsubs{W}
\endprooftree$$

The arrow elimination has become also a forall elimination. For the general case however it is not enough with the previous rule. We must consider bigger sums, which are not typable with such a rule, as well as arrows with more than one $\forall$, e.g.~$\forall X.\forall Y.(\U\to\R)$, where $\U\tsubs{V}\tsubs[Y]{W}$ has the correct type. Since it is under a sum, and the elimination must be done simultaneously in all the members of the sum, it is not possible with a traditional forall elimination.

The generalised arrow elimination as well as the rest of the typing rules are summarised in Fig.~\ref{fig:types}. Rules for the universal quantifier, axiom and introduction of arrow are the usual ones.
As discussed before, any sum of typable terms can be typed using rule~$+_I$. Notice that there is no elimination rule for~$+$ since the actual non-deterministic choice step (which eliminates one branch) is not considered here. For similar calculi where the elimination is present in the operational semantics, see e.g.~\cite{BucciarelliEhrhardManzonettoAPAL12,DiazcaroDowek12}. Finally, a rule assigns equivalent types to the same terms.

\begin{figure}[ht]
  \begin{center}
    $\prooftree
    \justifies\Gamma, x:\U\vdash x:\U
    \using ax
    \endprooftree$\qquad
    $\prooftree
    \justifies\Gamma\vdash\zero:\tzero
    \using ax_\tzero
    \endprooftree$
    \qquad
    $\prooftree\Gamma\vdash\t:\T\qquad\T\equiv\R
    \justifies\Gamma\vdash\t:\R
    \using\equiv
    \endprooftree$
    \medskip

    $\prooftree\Gamma, x:\U \vdash\t:\T
    \justifies\Gamma \vdash \lambda x.\t:\U\to \T
    \using\to_I
    \endprooftree$\qquad
    $\prooftree\Gamma \vdash\t:\sui{\alpha}\forall\vec X.(\U\to\T_i) \qquad \Gamma\vdash\u:\suj{\beta}\U\tsubs[\vec X]{\vec\V_j}
    \justifies\Gamma \vdash\t\u:\sui{\alpha}\suj{\beta}{\T_i\tsubs[\vec X]{\vec\V_j}}
    \using\to_E
    \endprooftree$
	\medskip

	$\prooftree \Gamma\vdash\t:\T \qquad \Gamma\vdash\u:\R
    \justifies \Gamma\vdash\t+\u:\T+\R
    \using +_I
    \endprooftree$
    \qquad
    $\prooftree\Gamma\vdash\t:\forall X.\U
    \justifies\Gamma\vdash\t:\U\tsubs{\V}
    \using\forall_E
    \endprooftree$\qquad
    $\prooftree\Gamma \vdash\t:\U\qquad X\notin FV(\Gamma)
    \justifies\Gamma \vdash\t:\forall X.\U
    \using\forall_I
    \endprooftree$
  \end{center}
  \caption{Typing rules of \add}
  \label{fig:types}
\end{figure}

\begin{example}\label{ex:flecha-elim}
  Let $\V_1=\U\tsubs{\W_1}$, $\V_2=\U\tsubs{\W_2}$, $\Gamma\vdash\v_1:\V_1$, $\Gamma\vdash\v_2:\V_2$, $\Gamma\vdash\lambda x.\t:\forall X.(\U\to\T)$ and $\Gamma\vdash\lambda y.\u:\forall X.(\U\to\R)$. Then
	$$\prooftree\Gamma\vdash\lambda x.\t+\lambda y.\u:\forall X.(\U\to\T)+\forall X.(\U\to\R)
        \qquad\Gamma\vdash\v_1+\v_2:\V_1+\V_2
	\justifies\Gamma\vdash(\lambda x.\t+\lambda y.\u)(\v_1+\v_2)\!:\T\tsubs{\W_1}+\T\tsubs{\W_2}+\R\tsubs{\W_1}+\R\tsubs{\W_2}
	\using\to_E
	\endprooftree$$
	Notice that this term reduces to
	   $\underbrace{(\lambda x.\t)\v_1}_{\T\tsubs{\W_1}} +
        \underbrace{(\lambda x.\t)\v_2}_{\T\tsubs{\W_2}} +
        \underbrace{(\lambda y.\u)\v_1}_{\R\tsubs{\W_1}} +
        \underbrace{(\lambda y.\u)\v_2}_{\R\tsubs{\W_2}}$.
\end{example}

\begin{example}\label{ex:identity}
 Let $\Gamma\vdash\v_1:\U$ and $\Gamma\vdash\v_2:\V$. Then the term $(\lambda x.x)(\v_1+\v_2)$, which reduces to $(\lambda x.x)\v_1+(\lambda x.x)\v_2$, can be typed in the following way:
$$\prooftree\Gamma\vdash\lambda x.x:\forall X.X\to X\qquad\Gamma\vdash\v_1+\v_2:\U+\V
\justifies\Gamma\vdash(\lambda x.x)(\v_1+\v_2):\U+\V
\using\to_E
\endprooftree$$
Notice that without the simultaneous forall/arrow elimination, it is not possible to type such a term.
\end{example}

\section{Main Properties}
\label{sec:sr-sn}
The \add\ type system is consistent, in the sense that typing is preserved by reduction (Theorem~\ref{thm:subjectreduction}).
Moreover, only terms with no infinite reduction are typable (Theorem~\ref{thm:strongnormalisation}).

The preservation of types by reduction, or \textit{subject reduction} property, is proved by adapting the proof of Barendregt~\cite[Section~4.2]{Barendregt92} for the System~$F$:
we first define a binary relation~$\preccurlyeq$ on types, and then prove the usual generation and substitution lemmas (cf.~Appendix~\ref{sec:pr-sr} for more details).

\begin{theorem}[Subject Reduction]\label{thm:subjectreduction}
  For any terms~$\t,\t'$, any context~$`G$ and any type~$\T$,
  if $\t\to^*\t'$ then
  $\Gamma\vdash\t\type
  \T\Rightarrow\Gamma\vdash\t'\type \T$.
\end{theorem}

We also prove the strong normalisation property (\ie no typable term has an infinite reduction) by adapting the standard method of \emph{reducibility candidates}~\cite[Chap.~14]{Girard89} to the \add\ type system.
The idea is to interpret types by reducibility candidates, which are
sets of strongly normalising terms.
Then we show that as soon as a term has a type, it is in its
interpretation, and thereby is strongly normalising.

We define here candidates as sets of \emph{closed} terms.
The set of all the closed terms is writen~\terms, and \sn\ denotes the set of
\emph{strongly normalising} closed terms.
In the following, we write \red{\t} for the set of reducts in one step
of a term~$\t$ (with any of the six rules given in
Sec.~\ref{sec:language}), and \red[*]{\t} for the set of its
reducts in any number of steps (including itself).
Both notations are naturally extended to sets of terms.
A term is a \emph{pseudo value} when it is an abstraction or a sum of them:
$\b, \b' ::= `lx.\t~|~\b+\b'.$
A term that is not a pseudo value is said to be \emph{neutral}, and we
denote by~\neu\ the set of \emph{closed neutral} terms.

\begin{definition}
  A set~$\sets``(=\terms$ is a \emph{reducibility candidate} if it
  satisfies the three following conditions:
  \cru Strong normalisation: $\sets``(=\sn$.
  \crd Stability under reduction: $\t`:\sets\ "=>"\ \red{\t}``(=\sets$.
  \crt Stability under neutral expansion:
    If $\t`:\neu$, then $\red{\t}``(=\sets$ implies $\t`:\sets$.
\end{definition}
We denote the reducibility candidates by $\cal A,B$, and the set of all the reducibility candidates by \cre.
Note that \sn\ is in \cre.
In addition, the term~$\zero$ is a neutral term with no reduct, so it is
in every reducibility candidate by~\crt.
Hence every reducibility candidate is non-empty.

Let \clo{\sets} be the closure of a set of terms~$\sets$ by \crt.
It can be defined inductively as follows:
If $\t`:\sets$, then $\t `: \clo \sets$, and if
$\t`:\neu$ and $\red{\t}``(=\clo\sets$, then ${\t `: \clo \sets}$.

We can actually use this closure operator to define reducibility candidates:
\begin{lemma}\label{lem:cre}
  If $\sets``(=\sn$, then $\clo{\red[*]{\sets}} `: \cre$.
\end{lemma}

In order to interpret types with reducibility candidates, we define
the operators `arrow', `plus' and `intersection' in \cre: Let $\mathcal{A,B}`:\cre$.
  We define: $\mathcal{A}\to \mathcal{B} =
    \{\t`:\terms/\ `A\u`:\mathcal{A}, \t\u`:\mathcal{B}\}$ and $\mathcal{A}\mas\mathcal{B} =
    \clo{(\mathcal{A+B})\cup\mathcal{A}\cup\mathcal{B}}$
     where $\mathcal{A+B}=\{\t+\u~/~\t`:\mathcal{A}\mbox{ and }\u`:\mathcal{B}\}$.
\begin{proposition}
\label{prop:op-cr}
  Let $\mathcal{A,B}`:\cre$.
  Then both~$\mathcal{A\to B}$ and~$\mathcal{A\mas B}$ are
  reducibility candidates.
  Moreover, if $(\mathcal{A}_i)_{i`:I}$ is a family of \cre,
  then~$\bigcap_{i`:I}\mathcal{A}_i$ is a reducibility candidate.
\end{proposition}

The operator~$+$ is commutative and associative on terms, and hence so is the operator~$+$ defined on sets of terms. Therefore, \mas\ is commutative and associative on reducibility candidates. In addition, $\clo{\emptyset}$ (a reducibility candidate according to Lemma~\ref{lem:cre}) is neutral with respect to \mas. Lemma~\ref{lem:assoc-mas} formalises these properties.
\begin{lemma}\label{lem:assoc-mas}
  Let~$\mathcal{A,B,C}`: \cre$.
  Then
    $\mathcal{A}\mas\mathcal{B} = \mathcal{B}\mas\mathcal{A}$,
    $(\mathcal{A}\mas\mathcal{B})\mas\mathcal{C} =
      \mathcal{A}\mas(\mathcal{B}\mas\mathcal{C})$ and
    $\mathcal{A}\mas\clo{\emptyset}=\mathcal{A}$.
\end{lemma}

Type variables are interpreted using \emph{valuations}, \ie
partial functions from type variables to reducibility candidates:
$`r~:=~\emptyset~|~`r,X\mapsto\mathcal{A}.$
The interpretation \itp{\T} of a type~$\T$ in a valuation~$`r$ (that is defined for each free type variable of~$\T$) is given by
\begin{displaymath}
  \begin{array}[t]{r@{\ =\ }l@{\qquad}r@{\ =\ }l}
    \itp{X} & `r (X) &
    \itp{\tzero} & \clo{\emptyset}\\
    \itp{\U \to \T} & \itp{\U} \to \itp{\T} &
    \itp{\T+\R} & \itp{\T}\mas\itp{\R}\\
    \itp{`AX. \T} &
    \bigcap_{\mathcal{A} `:\cre}\;\itp[`r,X\mapsto\mathcal{A}]{\T}\\
  \end{array}
\end{displaymath}
Lemma~\ref{lem:cre} and Proposition~\ref{prop:op-cr} ensure that each type is interpreted by
a reducibility candidate.
Furthermore, Lemma~\ref{lem:assoc-mas} entails that this
interpretation is well defined with respect to the type equivalences.
\begin{lemma}
  \label{lem:itp-equiv}
  For any types~$\T,\T'$, and any valuation~$`r$, if $\T\!\equiv\!\T'$ then $\itp{\T}\!=\!\itp{\T'}$.
\end{lemma}

\paragraph{Adequacy lemma.}
We show that this interpretation complies with typing judgements.
Reducibility candidates deal with closed terms, whereas
proving the adequacy lemma by induction requires the use of open terms
with some assumptions on their free variables (which are ensured by the context).
Therefore we use \emph{substitutions}~$`s$ to close terms:
  $$`s := \emptyset \;|\; x \mapsto u;`s \qquad \qquad
  \t_{\emptyset} = \t\quad , \quad
  \t_{x \mapsto u;`s} = \t\subs{u}_{`s}.$$

Given a context $`G$, we say that a substitution~$`s$
\emph{satisfies}~$`G$ for the valuation~$`r$ (notation:~\satis{`s}{`G})
when~$(x:\T) `: `G$ implies $`s(x) `: \itp{\T}$.
A typing judgement $`G\vdash\t\type \T$ is said to be \emph{valid}
(notation $`G\valid\t\type \T$) if for every valuation~$`r$, and for
every substitution~$`s$ satisfying~$`G$ for~$`r$, we have $\t_{`s} `: \itp{\T}$.
\begin{proposition}[Adequacy]
  \label{prop:adequacy}
  Every derivable typing judgement is valid:
  for each~$`G$, each term~$\t$ and each type~$\T$, we have that
  $ `G\vdash\t\type \T$ implies $`G\valid\t\type \T$.
\end{proposition}
This immediately provides the strong normalisation result:
\begin{theorem}[Strong normalisation]\label{thm:strongnormalisation}
  Every typable term in \add\ is strongly normalising.
\end{theorem}
\begin{proof}
  If a term~$\t$ is typable by a type~$\T$, then the adequacy lemma
  ensures that $\t`: \itp[\emptyset]{\T}$.
  As a reducibility candidate,~$\itp[\emptyset]{\T}$ is included
  in~\sn, and thus~$\t$ is strongly normalising. \qed
\end{proof}

\section{Logical Interpretation}
\label{sec:trans}
In this section, we interpret the \add\ type system into System $F$ with pairs (\F\ for short). Sum types are interpreted with Cartesian products. Since this product is neither associative nor commutative in \F, we first consider \add\ without type equivalences. This involves a slightly modified but equivalent type system, that we call \sadd. We then translate every term of \sadd\ into a term of \F. Finally, we show that our translation is correct with respect to typing in \add\ (Theorem~\ref{thm:type-trans}) and reduction (Theorem~\ref{theo:cor-red}).
\medskip

\paragraph{Structured Additive Type System.}
The system \sadd\ is defined with the same grammar of types as \add, and the same rules $ax$, $ax_{\tzero}$, $\to_I$, $+_I$, $`A_I$ and $`A_E$. There is no type equivalence, and thereby no commutativity nor associativity for sums (also $\tzero$ is not neutral for sums). Hence rule~$\to_E$, has to be precised. To specify what an $n$-ary sum is, we introduce the structure of trees for types.

\noindent
\parbox{0.843\textwidth}{
\begin{example}
  In \sadd, the type $(\U_1 + (\tzero+\U_2)) + \U_3$ is no longer equivalent to $\U_1+(\U_2+\U_3)$.
We can represent the first one by the labelled tree on the right.
\end{example}
}
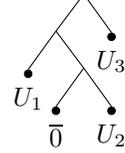
\parbox{0.15\textwidth}{
  \vspace{-1pt}
  \hfill\begin{tikzpicture}[scale = 0.7]
    \coordinate
    child {child {node[punto=\U_1]{}}
      child {child {node[punto=\tzero]{}}
                      child {node[punto=\U_2]{}}}}
                  child {node[punto=\U_3]{}};
  \end{tikzpicture}
}
\\
\parbox[b]{0.183\textwidth}
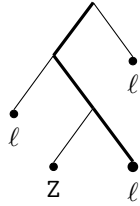
{
  \begin{tikzpicture}
    \coordinate
    child [very thick]{child [thin]{node[punto=\ell]{}}
      child {child [thin]{node[punto=\Az]{}}
                              child {node[punto=\ell]{}}}}
   child [thin]{node[punto=\ell]{}};
 \end{tikzpicture}
}
\parbox[b]{0.81\textwidth}{
  \quad
  To formalise \sadd, we use the standard representation of binary   trees, with some special leaves~$\ell$ (which can be labelled by a unit type):
\quad
\begin{math}
  \A,\A' := \ell~|~\Az~|~\As(\A,\A')~.
\end{math}

Each leaf is denoted by the finite word on the alphabet $\{\mathtt{l,r}\}$ (for \textit{left} and \textit{right}) representing the path from the root of the tree.
For instance, the type $(\U_1+(\tzero+\U_2))+\U_3$ is obtained using the labelling
\begin{math}
  \{\mathtt{ll}\mapsto \U_1,
  \mathtt{lrr}\mapsto \U_2,
  \mathtt{r}\mapsto \U_3
  \}
\end{math},
with the tree of the left.
}

We say that a labelling function~\lab\ (formally, a partial function from $\{\mathtt{l,r}\}^*$ to unit types) \emph{labels a tree~\A} when each of its leaves~$\ell$ is in the domain of~\lab.
In this case, we write $\A[\lab]$ the type of \sadd\ obtained by labelling~\A\ with~\lab.
Notice that conversely, for any type~$T$, there exists a unique tree~$\A_T$ and a labelling function~$\lab_T$ such that $T=\A_T[\lab_T]$.
The tree composition $\A`o\A'$ consists in ``branching'' $\A'$ to
each leaf~$\ell$ of~\A\ (\cf Example~\ref{ex:tree-compo} in Appendix~\ref{ap:ex-sadd}).
By extending the definition of labelling functions to functions from leaves to types, we have
$\A[w\mapsto\A'[\lab]] = \A`o\A'[wv\mapsto \lab(v)]$,
where~$w$ denotes a $\ell$-leaf of~\A, and~$v$ a $\ell$-leaf of $\A'$.
Then the rule for the arrow elimination in \sadd\ is:
$$\prooftree`G\vdash \t:\A[w\mapsto\forall\vec X.(\U\!\to T_w)]
			\qquad
			`G\vdash \u:\A'[v\mapsto \U\tsubs[\vec X]{\vec V_v}]
\justifies`G\vdash \t\u: \A`o\A'[wv\mapsto T_w\tsubs[\vec X]{\vec V_v}]
\using\estruct
\endprooftree$$

\noindent
where $wv$ is a word whose prefix~$w$ represents a leaf of \A~(cf.~Example~\ref{ex:tree-arrow-elim}).

\begin{proposition}[\add\ equivalent to \sadd]
  \label{prop:equiv-addstruct}~
    $`G\vdash \t: T$ is derivable in \add\ if and only if there is a type~$T'\equiv T$ such that $`G\vdash\t: T'$ is derivable in \sadd.
\end{proposition}
\paragraph{Translation into the System~$F$ with Pairs.}
We recall the syntax of \F~\cite{Dicosmo95}:
$$  \begin{array}
    {@{\hspace{40pt}}l@{\qquad}r@{\;}r@{\quad}l}
    \mathbf{Terms}: & t,u & := &
    x~|~`lx.t~|~tu~|~\star~|~\pair t u~|~\pil{t}~|~\pir{t}\\
\mathbf{Types}: & A,B & := & X~|~A \tof B~|~`AX.A~|~\uno~|~A`*B\\
\end{array}$$
(reduction and typing rules are well known, \cf Fig.~\ref{fig:sysF} on Appendix~\ref{ap:trad}).

In the same way than for the types, we define a term of~\F\ with a tree (whose binary nodes~\S\ are seen as pairs) and a partial function~$`t$ from $\{\mathtt{l,r}\}^*$ to $F_P$-terms.
We write~$`p_{`a_1\dots `a_n}(t)$ for~$`p_{`a_1}(`p_{`a_2}(\dots`p_{`a_n}(t)))$ (with $`a_i`:\{\mathtt{l,r}\}$).
Remark that if $t=\A[`t]$ and $w$ is a $\ell$-leaf of $\A$,
then~$`t(w)$ is a subterm of~$t$ that can be obtained by
reducing~$`p_{\overline w}(t)$, where~$\overline w$ is the mirror word of~$w$ (\cf~Example~\ref{ex:FPterms-trees}).

\noindent
\parbox[t]{.84\linewidth}{
  \begin{example}[Representation of $F_P$-terms with trees]
    \label{ex:FPterms-trees}\\
    Let~$t=\pair{\pair{u_1}{\pair{u_2}{u_3}}}{\star}$.
    Then~$t=\A[\mathtt{l\!l}\mapsto u_1,\mathtt{l\!r\!l}\mapsto u_2,
    \mathtt{l\!r\!r}\mapsto u_3]$
    (where \A\ is the tree on the right)
    and~$u_3$ reduces from~$`p_{221}(t)$.
  \end{example}
}
\parbox[t]{0.14\textwidth}
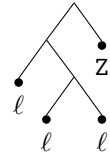
{
  \begin{tikzpicture}[baseline=10pt,scale = 0.7]
    \coordinate
    child { child {node[punto=\ell]{}}
            child {child {node[punto=\ell]{}}
                   child {node[punto=\ell]{}}}}
    child {node[punto=\Az]{}};
  \end{tikzpicture}
}

\vspace{-4ex}
Every type~$T$ is interpreted by a type $\trad{T}$ of \F.
$$\begin{array}{c}
\trad X = X, \quad
\trad{\tzero} = \uno,\quad
\trad{`AX.\U} = `AX.\trad \U,\quad
\\
\trad{\U\to T} = \trad \U\tof\trad T, \quad
\trad{T+R} =\trad T`*\trad R.
\end{array}$$

\noindent
Then any term~$\t$ typable with a derivation~$\D$ is interpreted by a $F_P$-term~$\tradt{\t}$:

\noindent If $\mathcal{D}=\dfrac{}{`G,x:T\vdash x:T}ax$, then $\tradt{x}=x$.
    \medskip

\noindent If $\mathcal{D}=\dfrac{}{`G\vdash\ve 0:\tzero}\,ax_{\tzero}$, then $\tradt{\ve 0}=\star$.
	\medskip

\noindent If $\mathcal{D}=\dfrac{\mathcal{D}_1\qquad\mathcal{D}_2}{`G\vdash\ve{\t+\u}:T+R}+_I$, then $\tradt{\t+\u}= \pair{\tradt[\mathcal{D}_1]{\t}}{\tradt[\mathcal{D}_2]{\u}}$.
	\medskip

\noindent If $\mathcal{D}=\dfrac{\mathcal{D'}}{`G\vdash`lx.\t:\U\to T}\to_I$, then $\tradt{`lx.\t}=`lx.\tradt[\mathcal{D'}]{\t}$.
	\medskip

\noindent If $\mathcal{D}=\dfrac{\mathcal{D}_1\qquad\mathcal{D}_2}{`G\vdash\t\u:\A`o\A'[wv\mapsto T_w\tsubs[\vec X]{\vec V_v}]}\!\estruct$,

	\hfill then $\tradt{\t\u}=\A`o\A'[wv\mapsto`p_{\overline w}(\tradt[\mathcal{D}_1]{\t})`p_{\overline v}(\tradt[\mathcal{D}_2]{\u})]$.

\noindent If $\mathcal{D}=\dfrac{\mathcal{D'}}{`G\vdash\t:`AX.\U}`A_I$, then $\tradt{\t}=\tradt[\mathcal{D'}]{\t}$.
	\medskip

\noindent If $\mathcal{D}=\dfrac{\mathcal{D'}}{`G\vdash \t:\U\tsubs{V}}`A_E$, then $\tradt{\t}=\tradt[\mathcal{D'}]{\t}$.
	\medskip

This interpretation is in fact a direct translation of sums by pairs
at each step of the derivation, except for the application:
informally, all the distributivity redexes are reduced before the translation of a term~\t\u, which requires to `know' the sum structure of~\t\ and~\u.
This structure is actually given by their type, and that is why we can only interpret typed terms.
\begin{example}
  If \t\ has type $(\U→\T_1)+(\U→\T_2)$ and \u\ has type $(\U+\tzero)+\U$, then we see them as terms of shape $\t_1+\t_2$ and $(\u_1+\zero)+\u_2$ respectively (the reducibility model of section~\ref{sec:sr-sn} ensures that they actually reduce to terms of this shape).
Indeed, the translation of \t\u\ reduces to the translation of
\linebreak[4]
\begin{math}
  \big(((\t_1\u_1)+\zero)+\t_1\u_2\big) +     \big(((\t_2\u_1)+\zero)+\t_2\u_2\big)
\end{math}:

\hfil\hfil
 \begin{math}
   \tradt{\t\u} =
   \pair{\
     \pair{\,
       \pair{t_1u_1}{\star}
       \,}{\,
       t_1u_2}
     \,}{\,
     \pair{\,
       \pair{t_2u_1}{\star}
       \,}{\,
       t_2u_2}
     }
 \end{math},
\\
where
$t_1=`p_{11}(\tradt[\mathcal{D}_1]{\t})$,\;$t_2=`p_{21}(\tradt[\mathcal{D}_1]{\t})$, $u_1=`p_{1}(\tradt[\mathcal{D}_2]{\u})$, and
$u_2=`p_{12}(\tradt[\mathcal{D}_2]{\u})$
\end{example}
\begin{theorem}[Correction with respect to typing]
  \label{theo:sadd-F}
  If a judgement~$`G\vdash\t:T$ is derivable in \sadd\ with
  derivation~$\mathcal{D}$, then $\trad{`G}\thesi\tradt{\t}:\trad{T}$.
\end{theorem}
The technical details for its proof are given in Appendix~\ref{ap:cor-trad}. In Appendix~\ref{ap:reverse} it is given a theorem showing that the translation is not trivial since it is reversible.

To return back to~\add, observe that if $\T\equiv\T'$, their translations are equivalent in \F\ (in the sense that there exists two terms establishing an isomorphism between them), and conclude with Proposition~\ref{prop:equiv-addstruct}.
\begin{theorem}
  \label{thm:type-trans}
  If a judgement~$`G\vdash\t:T$ is derivable in \add, then there is a   term~$t'$ of \F\ such that
$\trad{`G}\thesi t':\trad{T}$
\end{theorem}

To some extent, the translation from \sadd\ to \F\ is also correct with respect to reduction (technical details for its proof in Appendix~\ref{ap:cor-red}).
\begin{theorem}[Correction with respect to reduction]
  \label{theo:cor-red}
  Let $`G\vdash \t:T$ be derivable (by~\D) in \sadd, and $\t\to\u$.
  If the reduction is not due to rule $\t + \zero \to \t$,
  then there is \D' deriving $`G\vdash \u:T$, and
  \begin{math}
    \tradt{\t} \to^+\tradt[\D']{\u}.
  \end{math}
\end{theorem}

Notice that the associativity and commutativity of types have their analogous in the term equivalences.
However, the equivalence $T+0\equiv T$ has its analogous with a
reduction rule, $\t+\ve 0\to\t$. Since \sadd\ has no
equivalences, this reduction rule is not correct in the translation.
However, if $`G\vdash\t + \ve 0: T+\tzero$~is derivable by~\D\ in \sadd,
then there is some~$\D'= `G\vdash\t:T$ such that
$\varepsilon_{\trad{T+\tzero},\trad{T}} \tradt{\ve{t+0}}\to^*\tradt[\D']{\t}$, where $\varepsilon_{\trad{T+\tzero},\trad{T}}$ and
$\varepsilon_{\trad{T},\trad{T+\tzero}}$ are the terms establishing the isomorphism between
$\trad{T}$ and $\trad{T+\tzero}$ in \F.

\section*{Conclusion}
\label{sec:concl}
In this paper we considered an extension to call-by-value lambda calculus with a non-deterministic (or algebraic) operator $+$, and we mimiced its behaviour at the level of types. As we discussed in the introduction, this operator behaves like the algebraic sum with linear functions: $\mathbf f(x+y) = \mathbf f(x)+\mathbf f(y)$. However, our system is simulated by System~$F$ with pairs, which corresponds to the {\em non linear} fragment of \textsc{imell}.

This puts in the foreground the deep difference between the linearity in the algebraic sense (the one of Linear Logic), and the linearity of \add\ (which is the same, for instance, as Lineal~\cite{ArrighiDowekRTA08}). In the first case, a function is linear if it does not \textit{duplicate} its argument~$x$ (that is, $x^2$ --or $xx$-- will not appear during the computation), whereas in \add\ a linear behaviour is achieved by banning sum terms substitutions: while computing $(`lx.\t)(\u+\s)$, the argument $(\u+\s)$ will never be duplicated even if~\t\ is not linear in~$x$. We can only duplicate values (that intuitively correspond to constants in the algebraic setting, so their duplication does not break linearity).
Actually, in \add, the application is always distributed over the sum before performing the $`b$-reduction, and these both reductions do not interact. This is what our translation shows: all distributivity rules are simulated \textit{during} the translation (of the application), and then the $`b$-reduction is simulated in System~$F$, without paying any attention to the linearity.

As mentioned in the introduction, Lineal was meant for quantum computing and forcing the left distributivity is useful to prevent cloning. Moreover, it makes perfectly sense to consider any function as linear in this setting, since every quantum operator is given by a matrix, and thereby is linear. A \cbv\ reduction for this kind of calculus is thus entirely appropriate.


\paragraph{Acknowledgements.}
We would like to thank Olivier Laurent for the useful advice he gave us about the interpretation we present in this paper, as well as Pablo Arrighi for the fruitful discussions about Lineal and its linearity.

\bibliographystyle{splncs}
\bibliography{biblio}

\appendix
\section{Formalisation of the Proof of Subject Reduction}
\label{sec:pr-sr}
The preservation of types by reduction, or \textit{subject reduction} property, is proved by adapting the proof of Barendregt~\cite[Section~4.2]{Barendregt92} for the Sytem~$F$:
we first define a binary relation~$\preccurlyeq$ on types, and then we give the usual generation and substitution lemmas.
Finally, we give a needed property (Lemma~\ref{lem:type0val}) for the typing of $\ve 0$ and values.

\begin{definition}[Relation~$\preccurlyeq$ on types]\label{def:order}
  \begin{itemize}
  \item Given two types~$\U_1$ and $\U_2$, we write $\U_1\ssub \U_2$         if either
  \begin{itemize}
  \item $\U_2\equiv\forall X.\U_1$ or
  \item $\U_1\equiv\forall X.\U'$ and $\U_2\equiv \U'\tsubs{\T}$ for     some type~$T$.
  \end{itemize}
  \item We write $\preccurlyeq$ the reflexive (with respect to $\equiv$) transitive closure of $\ssub$.
  \end{itemize}
\end{definition}

\noindent
The following property says that if two arrow types are related by $\preccurlyeq$, then they are equivalent up to substitutions.

\begin{lemma}[Arrow comparison]\label{lem:arrowscomp}
  For any unit types $\U$, $\U'$ and types $\T$, $\T'$,
  if $\U'\to\T'\preccurlyeq \U\to\T$, then there exist $\vec\V,\vec X$ such that
  $\U\to \T\equiv(\U'\to \T)\tsubs[\vec X]{\vec\V}$.
\end{lemma}

As a pruned version of a subtyping system, we can prove the subsumption rule:
\begin{lemma}[$\preccurlyeq$-subsumption]
  \label{lem:ordertyping}
  For any context $\Gamma$, any term $\t$ and any unit types $\U$, $\U'$ such that $\U\preccurlyeq\U'$ and no free type variable in~$\U$ occurs in~$\Gamma$, if $\Gamma\vdash\t\type\U$ then $\Gamma\vdash\t\type\U'$.
\end{lemma}

Generation lemmas allows to study the conclusion of a derivation
so as to understand where it may come from, thereby decomposing the term in its
basic constituents.

\begin{lemma}[Generation lemmas]
  \label{lem:gen}
  For any context~$`G$, any terms~\t,\u, and any type~\T,
  \begin{enumerate}
  \item\label{it:genapp}
    $`G\vdash\t\u\type\T$ implies
    $`G\vdash\t\type\sui{n}\forall\vec X.(\U\to \T_i)$ and
    $`G\vdash\u\type\suj{m} \U\tsubs[\vec X]{\vec \V_j}$
    for some integers~$n$, $m$, some types~$\T_1,\dots,\T_n$, and some unit types~$\U,\vec \V_1,\dots,\vec \V_m$ such that
    $\sui{n}\suj{m}\T_i\tsubs[\vec X]{\vec\V_j}\preccurlyeq\T$.
  \item\label{it:genabs}
    $`G\vdash`lx.\t\type\T$ implies
    $`G,x\type\U\vdash\t\type\R$ for some types \U,\R\ such that $\U\to\R\preccurlyeq\T$.
  \item\label{it:gensum}
    $`G\vdash\t+\u\type\T$ implies
    $`G\vdash\t\type \R$ and
    $`G\vdash\u\type \S$ with for some types~\R, \S\ such that
    $\R+\S\equiv\T$.
  \end{enumerate}
\end{lemma}

\noindent
The following lemma is standard in proofs of subject reduction, and can be found for example in~\cite[Prop.~4.1.19]{Barendregt92} and \cite[Props.~8.2 and 8.5]{Krivine90}. It ensures than by substituting type variables for types or term variables for terms in an adequate manner, the type derived is still valid.

\begin{lemma}[Substitution]
  \label{lem:substitution}
  For any $`G$, $\T$, $\U$, $\v$ and $\t$,
  \begin{enumerate}
  \item $`G\vdash\t\type\T$ implies
    $`G\tsubs{\U}\vdash\t\type\T\tsubs{\U}$.
  \item If\ $`G,x\type\U\vdash\t\type\T$, and\
    $`G\vdash\v\type \U$, then\
    $\Gamma\vdash\t\subs{\v}\type\T$.
  \end{enumerate}
\end{lemma}

\noindent
Finally we need a property showing that $\zero$ is only typed by $\tzero$ and its equivalent types, and values are always typed by unit types or equivalent.

\begin{lemma}[Typing $\zero$ and values]
  \label{lem:type0val}
  \begin{enumerate}
  \item\label{it:type0}
    For any $\Gamma$, if
    $\Gamma\vdash\zero\type \T$\ then $\T\equiv\tzero$.
  \item\label{it:typeval}
    For any value~$\v$ (\ie a variable or an abstraction),
    if $`G\vdash\v\type\T$ then~\T\ is necessarily equivalent to a unit type.
  \end{enumerate}
\end{lemma}
\noindent Using all the previous lemmas, the proof of subject reduction is made by induction on typing derivation.

\section{Formalisation of the Translation into System~$F$}
\label{ap:trad}

\subsection{Some Examples}
\label{ap:ex-sadd}

\begin{example}[Tree composition]
  \label{ex:tree-compo}\\
  Let \quad \A=
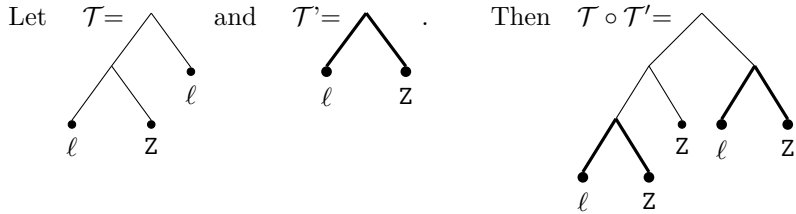
\hspace{-1cm}
  \begin{tikzpicture}[baseline={(0,-1ex)}]
    \coordinate
    child {child {node[punto=\ell]{}}
           child {node[punto=\Az]{}}}
    child {node[punto=\ell]{}};
  \end{tikzpicture}
  and \quad\A'=\hspace{-0.5cm}
  \begin{tikzpicture}[baseline={(0,-1ex)},very thick]
    \coordinate
    child {node[punto=\ell]{}}
    child {node[punto=\Az]{}};
  \end{tikzpicture}.
  \qquad Then\quad$\A`o\A'$=\hspace{-1.5cm}
  \begin{tikzpicture}[baseline={(0,-1ex)}]
    \tikzstyle{level 1}=[sibling distance=4em]
    \tikzstyle{level 2}=[sibling distance=2.5em]
    \tikzstyle{level 3}=[sibling distance=2.5em]
    \coordinate
    child {child {child [very thick]{node[punto=\ell]{}}
                  child [very thick]{node[punto=\Az]{}}}
           child {node[punto=\Az]{}}
         }
    child {child [very thick]{node[punto=\ell]{}}
          child [very thick]{node[punto=\Az]{}}};
  \end{tikzpicture}
\end{example}


\begin{example}[Arrow elimination rule in \sadd]
  \label{ex:tree-arrow-elim}\\
  The following derivation is correct:
  \begin{displaymath}
    \frac{`G\vdash\t:
      \big(`A\vec X.(\U\to T_1)+`A\vec X.(\U\to T_2)\big)+\tzero
      \quad
      `G\vdash\u:\U\tsubs[\vec X]{\vec V}+\tzero}{
      `G\vdash\t\u:
      \big((T_1\tsubs[\vec X]{\vec V}+\tzero)+
      (T_2\tsubs[\vec X]{\vec V}+\tzero)\big)+\tzero}
    \estruct
  \end{displaymath}
  Graphically, we can represent this rule as follows:\\
  if \t\ has type 
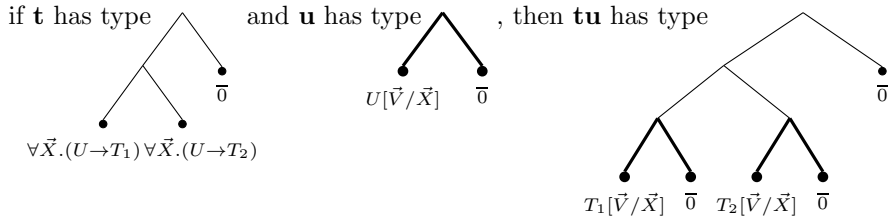
\hspace{-1.5cm}
  \begin{tikzpicture}[baseline={(0,-1ex)}]
    \coordinate
    child {child {node[punto=\hspace{-0.5cm}\scriptstyle\forall\vec X.(\U\to T_1)]{}}
           child {node[punto=\hspace{0.5cm}\scriptstyle \forall\vec X.(\U\to T_2)]{}}}
    child {node[punto=\scriptstyle\tzero]{}};
  \end{tikzpicture}
  \hspace{-0.5cm} and \u\ has type \hspace{-1.1cm}
  \begin{tikzpicture}[baseline={(0,-1ex)},very thick]
    \coordinate
    child {node[punto={\scriptstyle \U\tsubs[\vec X]{\vec V}}]{}}
    child {node[punto=\scriptstyle \tzero]{}};
  \end{tikzpicture},
  then \t\u\ has type \hspace{-2cm}
  \begin{tikzpicture}[baseline={(0,-1ex)}]
    \tikzstyle{level 1}=[sibling distance=6em]
    \tikzstyle{level 2}=[sibling distance=5em]
    \tikzstyle{level 3}=[sibling distance=2.5em]
    \coordinate
    child {child {child [very thick]{node[punto={\scriptstyle T_1\tsubs[\vec X]{\vec V}}]{}}
                  child [very thick]{node[punto=\scriptstyle\tzero]{}}}
           child {child [very thick]{node[punto={\scriptstyle T_2\tsubs[\vec X]{\vec V}}]{}}
                  child [very thick]{node[punto=\scriptstyle\tzero]{}}}}
    child{node[punto=\scriptstyle\tzero]{}}
    ;
  \end{tikzpicture}
\end{example}


\begin{figure}[ht]
$$
      \begin{array}[t]{c}
        \multicolumn{1}{l}{\bf Reduction\ rules:}
        \\ \displaystyle
        (`lx.t)u\to t\subs{u}
        \qquad;\qquad
        `p_i(\pair{t_1}{t_2})\to t_i
        \\ \displaystyle
        `lx.tx \to t \quad({\textstyle if\ x`;F\!V(t)\ })
        \qquad;\qquad
        \pair{\pil{p}}{\pir{p}}\to p
        \\ \\
        \multicolumn{1}{l}{\bf Typing\ rules:}
        \\ \displaystyle
        \frac{ }{`D,x:A\thesi x:A}{\scriptstyle Ax}
        \qquad;\qquad
        \frac{ }{`D\thesi \star:\uno}{\scriptstyle \uno}
        \qquad;\qquad
        \frac{`D,x:A\thesi t:B}{`D\thesi `lx.t:A\tof B}
        {\scriptstyle \tof I}
        \\ \\ \displaystyle
        \frac{`D\thesi t:A\tof B \quad `D\thesi u:A}{`D\thesi tu:B}
        {\scriptstyle \tof E}
        \qquad;\qquad
        \frac{`D\thesi t:A\quad `D\thesi u:B}{`D\thesi\pair{t}{u}:A`*B}
        {\scriptstyle `* I}
        \\ \\ \displaystyle
        \frac{`D\thesi t:A`*B}{`D\thesi\pil{t}:A}{\scriptstyle`*
          E_{\mathtt{l}}}
        \qquad;\qquad
        \frac{`D\thesi t:A`*B}{`D\thesi\pir{t}:B}{\scriptstyle`* E_{\mathtt{r}}}
        \\ \\ \displaystyle
        \frac{`D\thesi t:A\quad X`;F\!V(`D)}{`D\thesi t:`AX.A}
        {\scriptstyle `AI}
        \qquad;\qquad
        \frac{`D\thesi t:`AX.A}{`D\thesi t:A\tsubs{B}}{\scriptstyle `AE}
        \\ \\
      \end{array}
$$ 
  \caption{System $F$ with pairs}
  \label{fig:sysF}
\end{figure}


\subsection{Soundness with respect to Typing.}\label{ap:cor-trad}
We need first some lemmas and definitions.
It can be immediately checked that the tree structure of a type is
preserved by translation, as expressed in the following lemma.
\begin{lemma}\label{lem:tree-trad}
  If $T=\A[w\mapsto \U_w]$ is a type of \sadd, then $\trad{T} = \A[w\mapsto \trad{\U_w}].$
\end{lemma}

\begin{definition}
  \label{def:ftypes-tree}
  We call \textit{F-labelling} a function defined from leaves to types of \F.
  Given~$`f$, an F-labelling, and~$\A$, a tree, the type~$\A[`f]$ of \F\ is
  defined as expected:
$$\ell[`f] = `f(\varepsilon),\qquad
\Az[`f] = \uno, \qquad
\As(\A,\A')[`f] = \A[w\mapsto `f(\mathtt{l}w)]`*{\,\A'[w\mapsto `f(\mathtt{r}w)]\,}$$
\end{definition}
There is a trivial relation between the term-labelling of a
tree, and its F-labelling, that we give in the following lemma.
\begin{lemma}\label{lem:tree-type}
  Let~\A\ be a tree.
  \begin{enumerate}
  \item\label{it:constr}
    If $`G\thesi t_w: A_w$ for each
    $\ell$-leave~$w$, then
    $`G\thesi \A[w\mapsto t_w]:\A[w\mapsto A_w]$.
  \item\label{it:deconstr}
    If $`G\thesi t: \A[w\mapsto A_w]$, then for
    each~$\ell$-leaf of~\A, $`G\thesi`p_{\overline w}(t):A_w$.
  \end{enumerate}
\end{lemma}

\noindent
 \textbf{Theorem~\ref{theo:sadd-F} (Correction with respect to typing).}
{\it  If a judgement~$`G\vdash\t:T$ is derivable in \sadd\ with
  derivation~$\mathcal{D}$, then $\trad{`G}\thesi\tradt{\t}:\trad{T}$.}

\begin{proof}
  We prove this proposition by induction on the derivation~$\mathcal{D}$.
  If it ends with rule~$ax$ or~$ax_{\tzero}$, we use rule~$Ax$
  or~\uno\ respectively in \F.
  If the last rule of~$\mathcal{D}$ is~$+_I$ or~$\to_I$ we can
  conclude by induction.
  If the last rule is~$`A_I$, we just need to note that $X`;FV(`G)$
  implies $X`;FV(\trad{`G})$.
  If it is the rule~$`A_E$, we just have to note that
  $\trad{\,\U\tsubs{V}\,}=\trad{\U}\,\tsubs{\trad{V}\,}$ to conclude
  with induction hypothesis.
  The only interesting case is when~$\mathcal{D}$ ends with rule~$\estruct$:
  \begin{displaymath}
    \mathcal{D}=
    \frac{
      `G\vdash\t:\A[w\mapsto\forall\vec X.(\U\to T_w)]\qquad
      `G\vdash\u:\A'[v\mapsto \U\tsubs[\vec X]{\vec V_v}]}
    {`G\vdash\t\u:\A`o\A'[wv\mapsto T_w\tsubs[\vec X]{\vec V_v}]}
  \end{displaymath}
  By induction hypothesis,\ \
  $\trad{`G}\thesi\tradt[\mathcal{D}_1]{\t}:
  \trad{\A[w\mapsto\forall\vec X.(\U\to T_w)]}$\ \
  and\ \
  $\trad{`G}\thesi\tradt[\mathcal{D}_2]{\u}:
  \trad{\A'[v\mapsto \U\tsubs[\vec X]{\vec V_v}]}$.
  By Lemma~\ref{lem:tree-trad}, it means that
  $\trad{`G}\thesi\tradt[\mathcal{D}_1]{\t}:
  \A[w\mapsto\forall\vec X.\trad{\U}\tof \trad{T_w}]$\
  \ and\ \
  $\trad{`G}\thesi\tradt[\mathcal{D}_2]{\u}:
  \A'[v\mapsto\trad{\U}\tsubs[\vec X]{\vec{\trad{V_v}}}]$.
  By Lemma~\ref{lem:tree-type}.\ref{it:deconstr}, for
  every~$\ell$-leaf~$w$ of~\A, and every~$\ell$-leaf~$v$ of~\A', we
  can derive
  $$\prooftree
      \prooftree\trad{`G}\thesi `p_{\overline w}(\tradt[\mathcal{D}_1]{\t}):\forall\vec X.\trad{\U}\tof\trad{T_w}
      \justifies\trad{`G}\thesi `p_{\overline w}(\tradt[\mathcal{D}_1]{\t})\!:\trad{\U}\tsubs[\vec X]{\vec{\trad{V_v}}}\tof\trad{T_w}\tsubs[\vec X]{\vec{\trad{V_v}}}
      \endprooftree
      \
      \trad{`G}\thesi
      `p_{\overline v}(\tradt[\mathcal{D}_2]{\u})\!:\trad{\U}\tsubs[\vec X]{\vec{\trad{V_v}}}
   \justifies\trad{`G}\thesi
      `p_{\overline w}(\tradt[\mathcal{D}_1]{\t})~
      `p_{\overline v}(\tradt[\mathcal{D}_2]{\u}):
      \trad{T_w}\tsubs[\vec X]{\vec{\trad{V_v}}}
   \endprooftree$$
  Since
  \begin{math}
    \tradt{\t\u}=\A`o\A'
    [wv\mapsto
    `p_{\overline w}(\tradt[\mathcal{D}_1]{\t})~`p_{\overline v}
    (\tradt[\mathcal{D}_2]{\u})]
  \end{math},
  by Lemma~\ref{lem:tree-type}(\ref{it:constr}) we can conclude
  $\trad{`G}\thesi \tradt{\t\u}: \A`o\A'
    [wv\mapsto \trad{T_w}\tsubs[\vec X]{\vec{\trad{V_v}}}]$,
  and then conclude using Lemma~\ref{lem:tree-trad} again.\qed
\end{proof}


\subsection{Partial Translation from \F\ to \sadd.}\label{ap:reverse}

To show that the translation from \sadd\ to \F\ is meaningful and non trivial, we define a partial  encoding from \F\ to \sadd, and prove that it is the inverse of the previous translation.
We define inductively the partial function~$\tradi{\cdot}$ from the types of \F\ to those of~\sadd, as follows.
$$\tradi X = X\qquad\mbox{ and }\qquad\tradi{\uno} = \tzero\ ;$$
$$\mbox{if }\tradi{A}, \tradi{A'}\mbox{ and }\tradi{B}\mbox{ are defined, then}$$
$$\tradi{\forall X.A} = \forall X.\tradi A\mbox{ and }
\tradi{A\times B} = \tradi A+\tradi B\ ;$$
$$\mbox{ and if also }\tradi{A'}`:\U,\mbox{ then }
\tradi{A'\tof B} = \tradi{A'}\to\tradi B.$$

This translation is extended to contexts in the usual way.
Similarly, we define a partial function from terms of \F\ to those of \sadd:
\begin{displaymath}
  \tradti{x}=x\quad;\quad
  \tradti{\lambda x.t}=\lambda x.\tradti{t}\quad;\quad
  \tradti{tu}=\tradti{t}\tradti{u}\quad;\quad
  \tradti{\star}=\ve 0\quad;
\end{displaymath}
\begin{displaymath}
  \tradti{\A[wv\mapsto\pi_{\overline{w}}(t)\pi_{\overline{v}}(u)]}=\tradti{t}\tradti{u} \qquad \textrm{ if }\A\neq\Az\textrm{ and }\A\neq\ell\quad;
\end{displaymath}
\begin{displaymath}
  \tradti{\pair{t_1}{t_2}}=\tradti{t_1}+\tradti{t_2}\qquad
  \textrm{ if }   \pair{t_1}{t_2}\neq\A[wv\mapsto\pi_{\overline{w}}(u)\pi_{\overline{v}}(u')]
  \textrm{ for any }\A,u,u'
\end{displaymath}

This defines the inverse of~$\tradt{\cdot}$, as specified by the following theorem.
\begin{theorem}\label{thm:translation-no-trivial}
 If $\Gamma\vdash\t: T$ is derivable in \sadd\ with derivation \D, then
 $\tradi{\trad{\Gamma}}\vdash\tradti{\tradt{\t}}: \tradi{\trad{T}}$ is syntactically the same sequent.
\end{theorem}


\subsection{Soundness with respect to Reduction.}\label{ap:cor-red}
First we need a substitution lemma for the translation of terms.
\begin{lemma}\label{lem:substitutionoftranslation}
  Let $\D_1=\Gamma,x\type U\vdash\ve t\type T$ and
  $\D_2=\Gamma\vdash\ve{b}\type U$, then $\exists \D_3$ such that
  $\tradt[\D_1]{\ve t}\subs{\tradt[\D_2]{\ve b}}=
  \tradt[\D_3]{\ve t\subs{\ve b}}$.
\end{lemma}

\noindent
\textbf{Theorem~\ref{theo:cor-red} (Correction with respect to reduction). }
{\it  Let $`G\vdash \t:T$ be derivable (by~\D) in \sadd, and $\t\to\u$.
  If the reduction is not due to rule $\t + \zero \to \t$,
  then there is \D' deriving $`G\vdash \u:T$, and
  \begin{math}
    \tradt{\t} \to^+\tradt[\D']{\u}
  \end{math}.}

\begin{proof}
  The proof is long but straightforward using the previous lemmas. It follows by induction over $\D$. \qed
\end{proof}

\end{document}